\documentclass{article}
\usepackage[margin=2cm]{geometry}
\usepackage{amsmath,amssymb,amsfonts,amsthm}

\usepackage{textcomp}

\newcounter{NN}
\newtheorem{theorem}[NN]{Theorem}

\def\x{{\bf x}}
\def\f{{\bf f}}

\begin{document}

\title{Homogeneous Darboux polynomials \\
and generalising integrable ODE systems}
\author{
Peter H.~van der Kamp, D.I. McLaren and G.R.W. Quispel\\[2mm]
Department of Mathematics and Statistics, La Trobe University, Victoria 3086, Australia.\\
Email: P.vanderKamp@LaTrobe.edu.au\\[7mm]
Keywords: Darboux polynomials, Lotka-Volterra systems, Liouville integrability, superintegrability.
}

\maketitle

\begin{abstract}
We show that any system of ODEs can be modified whilst preserving its homogeneous Darboux polynomials. We employ the result to generalise a hierarchy of integrable Lotka-Volterra systems.
\end{abstract}

\section{Introduction}
We are concerned with systems of Ordinary Differential Equations (ODEs),
\begin{equation} \label{ode}
\dot{\x}=\f(\x),
\end{equation}
where $\dot \x$ denotes the time derivative of a vector $\x$. A Darboux polynomial (or second integral) of (\ref{ode}) is a polynomial $P(\x)$ such that $\dot{P}=C(\x)P$ for some function $C$ which is called the cofactor of $P$ \cite{G}. Darboux polynomials are important as the existence of sufficiently many Darboux polynomials implies the existence of a first integral, cf. Theorems 2.2 and 2.3 in \cite{G}. Recently their use was extended to the discrete setting in \cite{CEMOQTV}.

In this paper, we propose the following generalisation of any ODE system of the form (\ref{ode}):
\begin{equation} \label{gode}
\dot{\x}=\f(\x)+b(\x,t)\x,
\end{equation}
where $b$ is a scalar function of $\x,t$. We will prove that if $P$ is a homogeneous Darboux polynomial for (\ref{ode}), then $P$ is also a Darboux polynomial for (\ref{gode}) with a modified cofactor.

We show that in several examples the above generalisation preserves the integrability of the ODE, e.g. this is the case for generalisations of: (i) the 2-dimensional system
\begin{equation} \label{cex}
\begin{split}
\dot{x}&=x^2+2xy+3y^2,\\
\dot{y}&=2y(2x+y),
\end{split}
\end{equation}
found in \cite[Appendix]{C}, (ii) the 4-dimensional Lotka-Volterra (LV) system
\begin{equation} \label{OLV4}
\begin{split}
\dot{x}_1&=x_1(+x_2+x_3+x_4)\\
\dot{x}_2&=x_2(-x_1+x_3+x_4)\\
\dot{x}_3&=x_3(-x_1-x_2+x_4)\\
\dot{x}_4&=x_4(-x_1-x_2-x_3),
\end{split}
\end{equation}
as well as (iii) higher dimensional LV systems found in \cite{KKQTV}. For the LV systems we show that both Liouville integrability and superintegrability are preserved under certain generalisations given by (\ref{gode}).

\section{Darboux polynomials and integrals/integrability}
Note that if $P_1$ and $P_2$ are Darboux polynomials with cofactors $C_1$ and $C_2$ respectively, the product $P_1^aP_2^b$ is a Darboux polynomial with cofactor $aC_1+bC_2$. This implies that linear relations between cofactors give rise to integrals.

\vspace{5mm}

For the 2-dimensional system (\ref{cex}) three Darboux polynomials
\begin{equation} \label{DP123}
P_1=x+y,\qquad P_2=x-y,\qquad P_3=y,
\end{equation}
with cofactors given by
\begin{equation} \label{C123}
C_1=x+5y,\qquad C_2=x-y,\qquad C_3=4x+2y,
\end{equation}
respectively, were given in \cite[Example 2.21]{G}. As these cofactors satisfy the linear relation $C_1+3C_2-C_3=0$, an integral is given by
\[
I=P_1P_2^3P_3^{-1}=\frac{(x+y)(x-y)^3}{y}.
\]

\vspace{5mm}

The 4-dimensional LV system (\ref{OLV4}) admits linear Darboux polynomials
of the form
\[
P_{i,j}=\sum_{k=i}^j x_k, \text{ with } 1\leq i\leq j\leq 4,
\]
with corresponding cofactor
\[
C_{i,j}=-\sum_{k=1}^{i-1} x_k + \sum_{k=j+1}^n x_k.
\]
Because
\[
C_{1,2}-C_{3,3}+C_{4,4}=(x_3+x_4)-(-x_1-x_2+x_4)+(-x_1-x_2-x_3)=0,
\]
the rational function
\[
F=P_{1,2}P_{3,3}^{-1}P_{4,4}=(x_1+x_2)\frac{x_4}{x_3}
\]
is an integral. And similarly,
\[
C_{3,4}-C_{2,2}+C_{1,1}=(-x_1-x_2)-(-x_1+x_3+x_4)+(x_2+x_3+x_4)=0
\]
yields the rational integral
\[
G=P_{3,4}P_{2,2}^{-1}P_{1,1}=(x_3+x_4)\frac{x_1}{x_2}.
\]
As $C_{1,4}=0$, the function
\[
H=P_{1,4}=x_1+x_2+x_3+x_4
\]
provides a third integral. The functions $F,G,H$ are functionally independent, as their gradients are linearly independent, and therefore the LV system (\ref{OLV4}) is superintegrable. The variables $u_i=P_{1,i}$ provide a separation of variables, i.e. each variable satisfies the same differential equation $\dot{u}_i=u_i(H-u_i)$ which can be explicitly integrated, cf. \cite{B}

The system (\ref{OLV4}) is also a Hamiltonian system, with Hamiltonian $H$ and quadratic Poisson bracket, of rank 4,
\begin{equation} \label{bra}
\{x_i,x_j\}=x_ix_j,\qquad i<j.
\end{equation}
As both $F$ and $G$ Poisson commute with $H$, the systems $F,H$ and $G,H$, and hence the vector field (\ref{OLV4}), are Liouville integrable, cf. \cite{KKQTV}.

\section{Generalising ODE systems}
The following result is quite general, it generalises any ODE system (\ref{ode}) whilst preserving all homogeneous Darboux polynomials.

\begin{theorem} \label{GODEs}
Let $P(\x)$ be a homogeneous Darboux polynomial of degree $d$ with cofactor $C(\x)$ for the system of ODEs $\dot{\x}=f(\x)$. Then $P$ is a Darboux polynomial for the system $\dot{\x}=\f(\x)+b(\x,t)\x$, with cofactor $C+db(\x,t)$, where $b$ is a scalar function of $\x,t$.
\end{theorem}
\begin{proof}
As $P$ is homogeneous of degree $d$, we have $\x\cdot\nabla P=dP$. As $P$ is a Darboux polynomial for $\dot{\x}=f(\x)$, we have $\dot{P}=\nabla P\cdot \f=CP$. For the generalised system we then have
\[
\dot{P}=\nabla P\cdot (\f+b\x)= CP + bdP = (C+db)P.
\]
\end{proof}

\vspace{5mm}

We first apply Theorem \ref{GODEs} to the 2-dimensional system (\ref{cex}). With $b=ax+cy$ we obtain a generalisation of (\ref{cex}),
\begin{equation} \label{gcex}
\begin{split}
\dot{x}&=x^2+2xy+3y^2+(ax+cy)x,\\
\dot{y}&=2y(2x+y)+(ax+cy)y.
\end{split}
\end{equation}
Each $P_i$, $i=1,2,3$, given by (\ref{DP123}), is a linear Darboux polynomial for the system (\ref{gcex}) with modified cofactor $C^\prime_i=C_i+ax+cy$, where $C_i$ is given by (\ref{C123}).
As
\[
(c -a - 2)C^\prime_1-(a + c + 6)C^\prime_2+2(a + 1)C^\prime_3=0,
\]
the function
\[
K=P_1^{c -a - 2}P_2^{-(a + c + 6)}P_3^{2(a + 1)}=\frac{(x+y)^{c -a - 2}y^{2(a+1)}}{(x-y)^{a + c + 6}}.
\]
is a first integral of (\ref{gcex}).

\vspace{5mm}

Applying Theorem \ref{GODEs} to the 4-dimensional system (\ref{OLV4}), taking $b$ to be a constant, yields
\begin{equation} \label{NLV4}
\begin{split}
\dot{x}_1&=x_1(b+x_2+x_3+x_4)\\
\dot{x}_2&=x_2(b-x_1+x_3+x_4)\\
\dot{x}_3&=x_3(b-x_1-x_2+x_4)\\
\dot{x}_4&=x_4(b-x_1-x_2-x_3),
\end{split}
\end{equation}
whose Darboux polynomials $P_{i,j}$ now have cofactors $C^\prime_{i,j}=C_{i,j}+b$.
In particular, $H=P_{1,4}$ is no longer an integral, and the linear combinations
$C^\prime_{1,2}-C^\prime_{3,3}+C^\prime_{4,4}=C^\prime_{3,4}-C^\prime_{2,2}+C^\prime_{1,1}=b$ do not vanish. We have to subtract the cofactor $C^\prime_{1,4}=b$, which corresponds to dividing by $H$. This yields two integrals
\[
F^\prime=\frac{(x_1+x_2)x_4}{(x_1+x_2+x_3+x_4)x_3},
\]
and
\[
G^\prime=\frac{(x_3+x_4)x_2}{(x_1+x_2+x_3+x_4)x_1}.
\]
The new system (\ref{NLV4}) is still Hamiltonian, with the same bracket (\ref{bra}). The new Hamiltonian
\[
H^\prime=H-b\ln\left(\frac{x_1x_3}{x_2x_4}\right)
\]
is no longer rational. The integrals $F^\prime,G^\prime,H^\prime$ are functionally independent, and so the system (\ref{NLV4}) is superintegrable. Moreover, the functions $F^\prime$ and $G^\prime$ Poisson commute with $H^\prime$, hence the systems $F^\prime,H^\prime$ and $G^\prime,H^\prime$ are Liouville integrable. In the next section we generalise this example to arbitrary even dimensions.

\section{Integrability of a generalised $n$-dimensional LV system}
In \cite{KKQTV} the system of ODEs
\begin{equation} \label{OLV}
\dot{x}_i=x_i\left(\sum_{j>i} x_j - \sum_{j<i} x_j\right), \qquad i=1,\ldots,n,
\end{equation}
arose as a subsystem of the quadratic vector fields associated with multi-sums of products, and it was shown to be superintegrable as well as Liouville integrable. Integrable generalisations of the system (\ref{OLV}) have been obtained in \cite{CHK,EKV,KQV}. The generalisation
\begin{equation} \label{NLV}
\dot{x}_i=x_i\left(b + \sum_{j>i} x_j - \sum_{j<i}x_j\right), \qquad i=1,\ldots,n,
\end{equation}
of which (\ref{NLV4}) is a special case, seems to be new. In \cite{KKQTV} the LV system of ODEs (\ref{OLV}), with $n=2r$ even, was shown to admit the integrals, for $k=1,\ldots, r$,
\begin{equation} \label{FG}
\begin{split}
F_k&=
(x_1+x_2+\cdots+x_{2k})\frac{x_{2k+2}x_{2k+4}\cdots x_{n}}{x_{2k+1}x_{2k+3}\cdots x_{n-1}}, \\
G_k&=
(x_{n-2k+1}+x_{n-2k+2}+\cdots+x_{n})\frac{x_{1}x_{3}\cdots x_{n-2k-1}}{x_{2}x_{4}\cdots x_{n-2k}},
\end{split}
\end{equation}
The $n-1$ integrals $F_1,\ldots, F_{r-1}, G_1,\ldots, G_{r-1}, F_r=G_r=H=P_{1,n}$ were proven to be independent, and the sets
$
\{F_1,\ldots, F_{r-1}, H\}$, $\{G_1,\ldots, G_{r-1}, H\}
$
were proven to pairwise Poisson commute with respect to the bracket (\ref{bra}), which has rank $n$. Similar results were obtained for $n$ odd (here the rank of (\ref{bra}) is $n-1$), establishing the superintegrability as well as Liouville integrability of the $n$-dimensional LV system (\ref{OLV}) for all $n$. We consider a generalisation of the even-dimensional system.

\begin{theorem} \label{FT}
The system
\begin{equation} \label{ST}
\dot{x}_i=x_i(b + \sum_{j>i} x_j - \sum_{j<i} x_j), \qquad i=1,\ldots,n,
\end{equation}
where $n=2r$ is even, is both superintegrable and Liouville integrable.
\end{theorem}

\begin{proof}
The system is Hamiltonian with Hamiltonian
\[
H^\prime=H-bS, \text{ with } S=\ln\left(\frac{x_1x_3\cdots x_{n-1}}{x_2x_4\cdots x_n}\right).
\]

According to Theorem \ref{GODEs} the functions (\ref{FG}) and $H$ are Darboux functions (functions $F$ such that $\dot{F}=C(\x)F$ for some $C$) with cofactor $b$. Therefore, $n-2$ integrals are given by $F_i^\prime=F_i/H$ ,$G_i^\prime=G_i/H$, $i=1,\ldots,r-1$. Together with $H^\prime$ they form a set of
$n-1$ integrals,
\[
{\cal S}=\{F_1^\prime,\ldots, F_{r-1}^\prime, G_1^\prime,\ldots, G_{r-1}^\prime, H^\prime\},
\]
for which we will prove functional independence, thereby showing the superintegrability of (\ref{ST}). The trick is to add a function, $H$, and show that the bigger set ${\cal S} \cup \{H\}$ is functionally independent, by showing the determinant of the Jacobian to be non-zero, which is done using LU-decomposition, cf. \cite[Chapter 5]{Dinh}.
We may perform row operations, which we do by taking linear combinations of the functions $H$ and $H^\prime$ and ordering the functions in a particular way:
\[
Z=\big(2(H-H^\prime/2)/n^2,H/n^2,G^\prime_{n/2-1},F^\prime_1,
G^\prime_{n/2-2},F^\prime_2,\ldots,G^\prime_1,F^\prime_{n/2-1}\big).
\]
We then consider the scaled Jacobian $J=n^2\text{Jac}(Z)/2$ in the point $x_1=x_2=\cdots=x_n=b=1$. The first two functions in $Z$ are chosen so the first two rows in $J$ are given by $J_{i,j}=i+j+1 \mod 2$ $(i=1,2)$.

We conveniently introduce two sets of elementary functions
\[
P_{i,j}=x_i+x_{i+1}+\cdots+x_{j},\qquad
Q_{i,j}=x_i^{-1}x_{i+1}x_{i+2}^{-1}\cdots x_{j}^{(-1)^{j-i+1}},
\]
so that e.g. $F^\prime_k=\frac{P_{1,2k}Q_{2k+1,n}}{P_{1,n}}$.
As
\[
\frac{\partial F^\prime_k}{\partial x_i}=\begin{cases}
\frac{Q_{2k+1,n}}{P_{1,n}} - \frac{P_{1,2k}Q_{2k+1,n}}{P_{1,n}^2}  & i \leq 2k \\
-(-1)^i\frac{P_{1,2k}Q_{2k+1,n}}{x_iP_{1,n}} - \frac{P_{1,2k}Q_{2k+1,n}}{P_{1,n}^2}  & i > 2k,
\end{cases}
\]
we have
\[
\frac{n^2}{2}\frac{\partial F^\prime_k}{\partial x_i}\mid_{\bf x=\bf 1}=\begin{cases}
\frac{n}{2} - k  & i \leq 2k \\
-(-1)^i kn - k  & i > 2k.
\end{cases}
\]
In the point ${\bf 1}$ the gradient of $G_k$ is the gradient of $F_k$ read from right to left.
This yields, for $i>2$
\[
J_{i,j}=\begin{cases}
-((-1)^jn+1)(n-i+1)/2 & i\equiv 1,\ j<i \\
(i-1)/2 & i\equiv 1, j\geq i \\
(n-i+2)/2 & i\equiv 0, j< i-1 \\
((-1)^jn-1)(i-2)/2 & i\equiv 0,\ j\geq i-1,
\end{cases}
\]
where (here and in the sequel) the equivalence is taken modulo 2. Explicitly, for $n=10$ we have
\[
J={\small \begin{pmatrix}
1&0&1&0&1&0&1&0&1&0\\
0&1&0&1&0&1&0&1&0&1\\
36&-44&1&1&1&1&1&1&1&1\\
4&4&-11&9&-11&9&-11&9&-11&9\\
27&-33&27&-33&2&2&2&2&2&2\\
3&3&3&3&-22&18&-22&18&-22&18\\
18&-22&18&-22&18&-22&3&3&3&3\\
2&2&2&2&2&2&-33&27&-33&27\\
9&-11&9&-11&9&-11&9&-11&4&4\\
1&1&1&1&1&1&1&1&-44&36
\end{pmatrix}}.
\]

We define lower and upper triangular matrices
\[
L_{i,k}=\begin{cases}
M_{i,k} & k=1,2\\
1 & k=i\\
k/(n-k) & 1\equiv i = k+1, k>2 \\
0 & \text{otherwise},
\end{cases}
\quad
U_{k,j}=\begin{cases}
M_{k,j} & k=1,2\\
-n(n-k)/2 & k \equiv 1, j \equiv 1, j\geq k\\
n(n-k+2)/2 & k \equiv 1, j \equiv 0, j\geq k\\
-n^2/(n-k+1) & k \equiv 0, j \equiv 0, j\geq k\\
0 & k \equiv 0, j \equiv 0, j\geq k \text{ or } k>j.
\end{cases}
\]

When $n=10$ we have
\[
L={\small
\begin{pmatrix}
1&0&0&0&0&0&0&0&0&0\\
0&1&0&0&0&0&0&0&0&0\\
36&-44&1&0&0&0&0&0&0&0\\
4&4&\frac37&1&0&0&0&0&0&0\\
27&-33&0&0&1&0&0&0&0&0\\
3&3&0&0&1&1&0&0&0&0\\
18&-22&0&0&0&0&1&0&0&0\\
2&2&0&0&0&0&\frac73&1&0&0\\
9&-11&0&0&0&0&0&0&1&0\\
1&1&0&0&0&0&0&0&9&1
\end{pmatrix}},\
U={\small
\begin{pmatrix}
1&0&1&0&1&0&1&0&1&0\\
0&1&0&1&0&1&0&1&0&1\\
0&0&-35&45&-35&45&-35&45&-35&45\\
0&0&0&-{\frac{100}{7}}&0&-{\frac{100}{7}}&0&-{\frac{100}{7}}&0&-{\frac{100}{7}}\\
0&0&0&0&-25&35&-25&35&-25&35\\
0&0&0&0&0&-20&0&-20&0&-20\\
0&0&0&0&0&0&-15&25&-15&25\\
0&0&0&0&0&0&0&-{\frac{100}{3}}&0&-{\frac{100}{3}}\\
0&0&0&0&0&0&0&0&-5&15\\
0&0&0&0&0&0&0&0&0&-100
\end{pmatrix}}.
\]

We now show that $J=LU$, making use of the Kronecker delta, $\delta_{i,k}=1$ if $i=k$ and $0$ otherwise, and using summation over repeated indices. There are three cases:
\begin{itemize}
\item $i=1,2.$ We have $L_{i,k}=\delta_{i,k}$, so $L_{i,k}U_{k,j}=U_{i,j}=M_{i,j}$.
\item $1\equiv i>2$. We have
$\begin{aligned}[t]
L_{i,k}U_{k,j}&=(n-1)(n-i+1)U_{1,j}/2-(n+1)(n-i+1)U_{2,j}/2+U_{i,j}\\
&=\begin{cases}
-((-1)^jn+1)(n-i+1)/2 & i>j \\
(n-1)(n-i+1)/2-n(n-i)/2=(i-1)/2 & 1 \equiv j \geq i \\
-(n+1)(n-i+1)/2+n(n-i+2)/2=(i-1)/2 & 0 \equiv j \geq i.
\end{cases}
\end{aligned}$
\item $0\equiv i>2$. We have
$\begin{aligned}[t]
L_{i,k}U_{k,j}&=(n-i+2)(U_{1,j}+U_{2,j})/2+(i-1)U_{i-1,j}/(n-i+1)+U_{i,j}\\
&=\begin{cases}
\frac{n-i+2}{2} & j<i-1 \\
\frac{n-i+2}{2}-\frac{(i-1)n(n-i+1)}{2(n-i+1)}=-\frac{(i-2)(n+1)}2 & 1 \equiv j \geq i \\
\frac{n-i+2}{2}+\frac{(i-1)n(n-i+3)}{2(n-i+1)}-\frac{n^2}{n-i+1}=\frac{(i-2)(n-1)}2 & 0 \equiv j \geq i.
\end{cases}
\end{aligned}$
\end{itemize}
As both $L$ and $U$ have non-zero diagonal elements, the determinant of $J$ is non-zero. Hence the set $S$ is functionally independent. This shows that (\ref{ST}) is superintegrable.

Next we prove that each pair of functions in the set
$
\{F_1^\prime,\ldots, F_{r-1}^\prime, H^\prime\}$
Poisson commutes with respect to the bracket (\ref{bra}). Due to the Leibniz rule, the brackets $\{F_i^\prime,F_j^\prime\}=\{F_i/H,F_j/H\}$, with $1\leq i,j < r$, can be expressed in terms of $\{F_i,F_j\}$, $\{F_i,H\}$, $\{H,F_j\}$, which all vanish. We also have $\{F_i^\prime,H^\prime\}=0$ as the $F_i^\prime$ are integrals and $H^\prime$ is the Hamiltonian function of the system. Similarly, it follows that the functions in $\{G_1^\prime,\ldots, G_{r-1}^\prime, H^\prime\}$ Poisson commute. This shows that (\ref{ST}) is Liouville integrable.
\end{proof}

{\bf Remark 1.} Similar to the above, one can also show that the system
\begin{equation} \label{PST}
\dot{x}_i=x_i\left(b(S) + \sum_{j>i} x_j - \sum_{j<i} x_j\right), \qquad i=1,\ldots,n=2r,
\end{equation}
where $b$ is an arbitrary integrable function, is both superintegrable and Liouville integrable. The system (\ref{PST}) is a Hamiltonian system with Hamiltonian $H^\ast=H-B(S)$, where $B$ is the anti-derivative of $b$.\\

{\bf Remark 2.} In general, the $b$-generalisation (\ref{gode}) of a Hamiltonian system (\ref{ode}) will not be Hamiltonian. We hope to discuss some other cases in which the generalisation is Hamiltonian in a future publication. The reason that we have restricted the dimension of the Lotka-Volterra systems (\ref{ST}) to be even is that it seems unclear whether
a Hamiltonian exists in the general odd-dimensional case.

\section{Acknowledgements}
The authors thank Pambos Evripidou for careful reading of the manuscript and his suggestions to include a generalisation of (\ref{cex}) as well as remark 1. GRWQ is grateful to the Simons Foundation for a Fellowship during the early stages of this work, and acknowledges support from the European Union Horizon 2020 research and innovation programmes under the Marie Sk\l{}odowska-Curie grant agreement No. 691070.


\begin{thebibliography}{10}
\bibitem{B}
O.I. Bogoyavlenskij, Integrable Lotka-Volterra systems, Regul. Chaotic Dyn. {\bf 13} (2008) 543--556.

\bibitem{CEMOQTV}
E. Celledoni, C. Evripidou, D.I. McLaren, B. Owren, G.R.W. Quispel, B.K. Tapley and P.H. van der Kamp, Using discrete Darboux polynomials to
detect and determine preserved measures
and integrals of rational maps, J. Phys. A: Math. Theor. {\bf 52} (2019) 31LT01 (11pp).

\bibitem{CHK}
H. Christodoulidi, A.N.W. Hone, T.E. Kouloukas. A new class of integrable Lotka-Volterra systems. J. Comput. Dyn. {\bf 6} (2) (2019) 223--237.

\bibitem{C}
C.B. Collins, Algebraic conditions for a centre or a focus in some simple systems of arbitrary degree,
J. Math. Anal. Appl. {\bf 195} (1995) 719--735.

\bibitem{EKV}
C. Evripidou, P. Kassotakis, P. Vanhaecke, Integrable reductions of the dressing chain, J. Comput. Dyn. {\bf 6} (2) (2019) 277--306.

\bibitem{G}
A. Goriely, Integrability and Nonintegrability of Dynamical Systems, World Scientific, 2001 (436pp).

\bibitem{KQV}
T.E. Kouloukas, G.R.W. Quispel and P. Vanhaecke, Liouville integrability and superintegrability of a generalized Lotka-Volterra system and its Kahan discretization, J. Phys. A {\bf 49} (2016) 13pp.

\bibitem{Dinh}
D.T. Tran, Complete integrability of maps obtained as reductions of integrable lattice equations, PhD thesis, La Trobe University, Australia, 2011.

\bibitem{KKQTV}
P.H. van der Kamp, T.E. Kouloukas, G.R.W. Quispel, D.T. Tran and P. Vanhaecke, Integrable and superintegrable systems associated with multi-sums of products, Proc. R. Soc. Lond. Ser. A Math. Phys. Eng. Sci. {\bf 470} (2014) 20140481.



\end{thebibliography}
\end{document}